\documentclass[11pt]{article}

\usepackage[margin=1in]{geometry}
\usepackage[T1]{fontenc}
\usepackage[utf8]{inputenc}
\usepackage{lmodern}
\usepackage{amsmath,amssymb,amsthm,mathtools}
\usepackage{microtype}
\usepackage{hyperref}
\usepackage{xcolor}
\usepackage{array}
\usepackage{booktabs}
\usepackage{enumitem}

\hypersetup{
  colorlinks=true,
  linkcolor=blue!60!black,
  citecolor=blue!60!black,
  urlcolor=blue!60!black
}

\newtheorem{lemma}{Lemma}

\newtheorem{theorem}[lemma]{Theorem}
\newtheorem{corollary}[lemma]{Corollary}

\newcommand{\F}{\mathbb{F}}
\newcommand{\supp}{\operatorname{supp}}
\newcommand{\wt}{\operatorname{wt}}
\newcommand{\Id}{I}
\newcommand{\X}{X}
\newcommand{\Y}{Y}
\newcommand{\Z}{Z}

\newcommand{\calL}{\mathcal{L}}

\newcommand{\conf}{\operatorname{conf}}
\newcommand{\anti}{\operatorname{anti}}

\title{{\fontfamily{ppl}\selectfont \textbf{Counting anticommuting Pauli pairs in linear time}}}
\author{
    Hyunho Cha and Jungwoo Lee\\
    \small NextQuantum and Department of Electrical and Computer Engineering\\
    \small Seoul National University, Seoul 08826, Republic of Korea\\
    \small \texttt{\{ovalavo, junglee\}@snu.ac.kr}
}
\date{}

\begin{document}
\maketitle

\begin{abstract}
Many quantum computing workflows manipulate long lists of Pauli strings.
A basic classical subroutine involves taking $m$ Pauli strings on $n$ qubits, each of weight bounded by a constant, to determine if they are pairwise commuting, identify any counterexamples, or calculate the exact number of anticommuting unordered pairs.
The standard general-purpose route represents Pauli strings in binary symplectic form and checks pairs in $O(m^2)$ time.
Here, we provide an $O(m)$ algorithm for the bounded locality regime.
It maintains counts of all labeled subpatterns of previously inserted strings and answers each new string query by a subset zeta identity.
Our algorithm is particularly useful for processing large collections of Pauli strings within the bounded locality regime.
\end{abstract}


\section{Introduction}

Quantum states are vectors, and quantum gates are matrices. A great deal of quantum algorithm design is nevertheless written in a much smaller symbolic language consisting of products of the four one-qubit matrices
$
\Id,
$
$
\X=(\begin{smallmatrix}0&1\\1&0\end{smallmatrix}),
$
$
\Y=(\begin{smallmatrix}0&-i\\ i&0\end{smallmatrix}),
$
and
$
\Z=(\begin{smallmatrix}1&0\\0&-1\end{smallmatrix}).
$
A tensor product such as $\X\otimes \Id\otimes \Z\otimes \Y$ is called a Pauli string. These strings form an operator basis, so Hamiltonians, observables, stabilizer checks, Pauli rotations, and many intermediate objects in quantum software can be stored as weighted sums of Pauli strings \cite{aaronson2004improved, nielsen2010quantum}.

Consider the task of local Pauli commutation certification and anticommutation counting, which is a common classical subroutine within quantum algorithms and software. Given an input of $m$ sparse $n$-qubit Pauli strings $P^{(1)},\dots,P^{(m)}$ where each string acts nontrivially on at most $k$ qubits, the objective is to produce one of several outputs based on the specific requirements of the workflow. Specifically, the routine should be able to calculate the exact number of unordered anticommuting pairs $\{i,j\}$, provide a certificate that all strings commute pairwise, or, in cases where they do not all commute, identify a witness pair $(i,j)$ such that $P^{(i)}P^{(j)}=-P^{(j)}P^{(i)}$.

This primitive is important because simultaneous measurement and simultaneous diagonalization require commutation. Variational quantum eigensolver measurement reduction partitions Pauli terms into commuting families, frequently through graph coloring or clique cover formulations \cite{gokhale2020n, hamamura2020efficient, verteletskyi2020measurement, yen2020measuring}. Recent partitioning methods avoid some graph costs for dense operator bases but explicitly target dense settings rather than arbitrary sparse local Hamiltonians \cite{reggio2024fast}. Pauli-based computation uses sequences of pairwise commuting Pauli measurements \cite{peres2023quantum}. Modern Pauli arithmetic libraries optimize multiplication and commutator testing by bit operations, but the batching problem still appears naturally whenever many pairs must be screened or many candidate commuting sets must be verified \cite{muller2026pauliengine}.

Listing all anticommuting edges can take $\Omega(m^2)$ time when the output graph is dense. Exact counting and certification, however, have only a logarithmic-size or constant-size output and are exactly what many verification steps need. In this work, we improve the pairwise graph baseline from quadratic to linear in $m$ when $k$ is fixed.

\section{Preliminaries}

\subsection{Qubits, tensor products, and gates}

A single qubit is a unit vector in the two-dimensional complex vector space $\mathbb{C}^2$. The computational basis is
\[
    |0\rangle=\begin{pmatrix}1\\0\end{pmatrix},
    \qquad
    |1\rangle=\begin{pmatrix}0\\1\end{pmatrix}.
\]
An $n$-qubit pure state is a unit vector in the tensor product space
\[
    (\mathbb{C}^2)^{\otimes n}
    = \underbrace{\mathbb{C}^2\otimes\cdots\otimes\mathbb{C}^2}_{\times n},
\]
which has dimension $2^n$. A quantum gate on $n$ qubits is a $2^n\times 2^n$ unitary matrix, meaning a complex matrix $U$ satisfying $U^\dagger U=UU^\dagger=I_{2^n}$.

When matrices $A$ and $B$ act on two different qubit registers, their combined action is described by the tensor product $A\otimes B$. For matrices $A,B,C,D$ of compatible sizes,
\[
    (A\otimes B)(C\otimes D)=(AC)\otimes(BD).
\]
The same identity extends to any number of tensor factors.

\subsection{Pauli letters and Pauli strings}

The one-qubit Pauli letters are
$
\Id,
$
$
\X=(\begin{smallmatrix}0&1\\1&0\end{smallmatrix}),
$
$
\Y=(\begin{smallmatrix}0&-i\\ i&0\end{smallmatrix}),
$
and
$
\Z=(\begin{smallmatrix}1&0\\0&-1\end{smallmatrix}).
$
They satisfy
\[
  \X^2=\Y^2=\Z^2=\Id,
\]
and any two distinct non-identity Pauli matrices anticommute:
\[
  \X\Y=-\Y\X,
  \qquad
  \Y\Z=-\Z\Y,
  \qquad
  \Z\X=-\X\Z.
\]
The identity commutes with all three.

A phase-free $n$-qubit Pauli string is a word
\[
        P=(P_1,\ldots,P_n)\in\{\Id,\X,\Y,\Z\}^n,
\]
identified with the matrix
\[
        P_1\otimes P_2\otimes\cdots\otimes P_n.
\]
The support and weight of $P$ are
\[
   \supp(P)=\{j\in\{1,\ldots,n\}:P_j\neq\Id\},
   \qquad
   \wt(P)=|\supp(P)|.
\]
The string is $k$-local if $\wt(P)\leq k$. Global phases $\pm1,\pm i$ are ignored because they do not affect whether two strings commute.

Sparse input means that $P$ is stored as the list of pairs $(j,P_j)$ for $j\in\supp(P)$, not as all $n$ letters.

\subsection{Commutation and anticommutation}

Two square matrices $A$ and $B$ commute if $AB=BA$ and anticommute if $AB=-BA$. For Pauli strings, exactly one of these two possibilities holds.

For two Pauli strings $P,Q$, define their conflict set
\[
    \conf(P,Q)=\{j: P_j\neq\Id,
                   \ Q_j\neq\Id,
                   \ P_j\neq Q_j\}.
\]
This is the set of qubits where both strings act nontrivially and the two one-qubit Pauli letters are different.

\begin{lemma}
\label{lem:conflict-sign}
For any two phase-free $n$-qubit Pauli strings $P,Q\in\{\Id,\X,\Y,\Z\}^n$,
\[
      PQ=(-1)^{|\conf(P,Q)|}QP.
\]
Consequently, $P$ and $Q$ commute if $|\conf(P,Q)|$ is even, and anticommute if $|\conf(P,Q)|$ is odd.
\end{lemma}

\subsection{The binary symplectic baseline}

A standard representation maps each Pauli string to two binary vectors $x,z\in\F_2^n$ by
\[
  \Id\leftrightarrow(0,0),\qquad
  \X\leftrightarrow(1,0),\qquad
  \Z\leftrightarrow(0,1),\qquad
  \Y\leftrightarrow(1,1).
\]
If $P\leftrightarrow(x,z)$ and $Q\leftrightarrow(x',z')$, then
\[
   P\text{ anticommutes with }Q
   \quad\Longleftrightarrow\quad
   x\cdot z' + z\cdot x' =1\pmod 2.
\]
This is the binary symplectic form. It gives a very fast single-pair test, especially with word-level bit operations \cite{aaronson2004improved, muller2026pauliengine}. For a list of $m$ strings, however, applying this test to every unordered pair costs $\Theta(m^2)$ pair checks. If the strings are sparse and $k$-local, each pair can be checked in $O(k)$ time by merging or hashing their supports, for total $O(m^2k)$.

\section{Problem formulation}

Fix positive integers $n,m,k$. The input is a sequence
\[
    P^{(1)},\ldots,P^{(m)}\in\{\Id,\X,\Y,\Z\}^n
\]
with $\wt\!\big(P^{(i)}\big)\leq k$ for every $i$, given in sparse form.

The anticommutation count is
\[
  A\big(P^{(1)},\ldots,P^{(m)}\big)
   = \left|\left\{\{i,j\}:1\leq i<j\leq m,
      \ \left|\conf\!\big(P^{(i)},P^{(j)}\big)\right|\equiv 1\pmod2\right\}\right|.
\]
The certification problem is to decide whether $A\big(P^{(1)},\ldots,P^{(m)}\big)=0$. The witness version must return a pair $(i,j)$ with $i<j$ and $P^{(i)}P^{(j)}=-P^{(j)}P^{(i)}$ when such a pair exists.
If the requested output is the full anticommutation graph, then the output itself can have $\Theta(m^2)$ edges. Our subquadratic result is for counting, certification, and witness finding, not for explicitly listing all edges.

\section{The locality-zeta data structure}

\subsection{Patterns}

A labeled Pauli pattern is a pair $(A,a)$, where $A\subseteq\{1,\ldots,n\}$ is a finite set of qubit positions and
\[
      a:A\to\{\X,\Y,\Z\}
\]
assigns a non-identity Pauli letter to every position in $A$. The empty pattern is $(\varnothing,\varnothing)$.

A Pauli string $Q$ contains the pattern $(A,a)$, written
\[
      Q\succeq(A,a),
\]
if $A\subseteq\supp(Q)$ and $Q_j=a(j)$ for every $j\in A$.

For a multiset $G$ of previously inserted Pauli strings, define the table of pattern counts
\[
      D_G(A,a)=|\{Q\in G: Q\succeq(A,a)\}|.
\]
The table stores counts for patterns that occur inside the inserted strings. If a key has never been inserted, its count is interpreted as zero. In particular,
\[
      D_G(\varnothing,\varnothing)=|G|.
\]

\subsection{Insertion}

To insert a $w$-local Pauli string $Q$, enumerate every subset $A\subseteq\supp(Q)$ and increment
\[
      D_G(A,Q|_A)
\]
by one, where $Q|_A$ is the letter assignment on $A$.

\subsection{Query}

Let $P$ be a query string with support $S=\supp(P)$. For a subset $A\subseteq S$, define
\[
      \calL_P(A)=\{a:A\to\{\X,\Y,\Z\}:a(j)\neq P_j\text{ for every }j\in A\}.
\]
Thus $\calL_P(A)$ contains the assignments that force a conflict with $P$ on every position of $A$.

Define
\[
      F_G(P,A)=\sum_{a\in\calL_P(A)}D_G(A,a).
\]
This is the number of previously inserted strings that conflict with $P$ on every coordinate in $A$, with no restriction on other coordinates.

The query computes
\[
      Z_G(P)=\sum_{A\subseteq S}(-2)^{|A|}F_G(P,A)
\]
and returns
\[
      \anti_G(P)=\frac{|G|-Z_G(P)}{2}.
\]
The number $\anti_G(P)$ will be proved to be the number of strings $Q\in G$ that anticommute with $P$.

\subsection{Pseudocode}

The following pseudocode uses a dictionary $D$ whose keys are labeled patterns. The function \textsc{ConflictingAssignments}$(P,A)$ enumerates the assignments $a$ on $A$ with $a(j)\neq P_j$ for each $j\in A$.

\medskip
\noindent\fbox{%
\begin{minipage}{0.94\linewidth}
\textbf{Procedure \textsc{Insert}$(Q,D)$}
\begin{enumerate}[leftmargin=2.2em]
\item For every subset $A\subseteq\supp(Q)$:
\begin{enumerate}[leftmargin=2.0em]
\item Let $a=Q|_A$.
\item Set $D[(A,a)]\leftarrow D[(A,a)]+1$.
\end{enumerate}
\end{enumerate}

\textbf{Procedure \textsc{AntiCountAgainstPrevious}$(P,D,N)$}
\begin{enumerate}[leftmargin=2.2em]
\item Set $Z\leftarrow0$.
\item For every subset $A\subseteq\supp(P)$:
\begin{enumerate}[leftmargin=2.0em]
\item Set $F\leftarrow0$.
\item For every $a\in\textsc{ConflictingAssignments}(P,A)$:
\begin{enumerate}[leftmargin=1.8em]
\item Set $F\leftarrow F+D[(A,a)]$.
\end{enumerate}
\item Set $Z\leftarrow Z+(-2)^{|A|}F$.
\end{enumerate}
\item Return $(N-Z)/2$.
\end{enumerate}

\textbf{Procedure \textsc{CountAllAnticommutingPairs}$\big(P^{(1)},\ldots,P^{(m)}\big)$}
\begin{enumerate}[leftmargin=2.2em]
\item Initialize an empty dictionary $D$, set $N\leftarrow0$, and set $T\leftarrow0$.
\item For $i=1,\ldots,m$:
\begin{enumerate}[leftmargin=2.0em]
\item Set $c\leftarrow\textsc{AntiCountAgainstPrevious}\big(P^{(i)},D,N\big)$.
\item Set $T\leftarrow T+c$.
\item Run \textsc{Insert}$\big(P^{(i)},D\big)$.
\item Set $N\leftarrow N+1$.
\end{enumerate}
\item Return $T$.
\end{enumerate}
\end{minipage}}

\medskip
For certification with a witness, run the same loop. If $c>0$ for the current $P^{(i)}$, scan the previous strings $P^{(1)},\ldots,P^{(i-1)}$ using the elementary parity test from Lemma~\ref{lem:conflict-sign} until an anticommuting pair is found. This scan costs $O(mk)$ additional time.

\section{Correctness}

\begin{lemma}
\label{lem:FG-meaning}
Let $G$ be the multiset of strings inserted into the dictionary. Let $P$ be a query string and let $A\subseteq\supp(P)$. Then
\[
      F_G(P,A)=|\{Q\in G:A\subseteq\conf(P,Q)\}|.
\]
\end{lemma}

\begin{proof}
By definition,
\[
      F_G(P,A)=\sum_{a\in\calL_P(A)}D_G(A,a).
\]
For a fixed assignment $a\in\calL_P(A)$, the value $D_G(A,a)$ counts exactly the strings $Q\in G$ such that $Q_j=a(j)$ for all $j\in A$. Since $a\in\calL_P(A)$, every such $a(j)$ is a non-identity Pauli letter different from $P_j$. Also $j\in A\subseteq\supp(P)$, so $P_j\neq\Id$. Therefore each counted $Q$ has $j\in\conf(P,Q)$ for every $j\in A$, which means $A\subseteq\conf(P,Q)$.

Conversely, suppose $Q\in G$ satisfies $A\subseteq\conf(P,Q)$. Then, for every $j\in A$, the letter $Q_j$ is non-identity and different from $P_j$. Hence the assignment $a=Q|_A$ lies in $\calL_P(A)$, and $Q$ is counted by the summand $D_G(A,a)$. A string cannot be counted by two different assignments on the same set $A$, because its restriction $Q|_A$ is unique. Thus the sum counts exactly the strings $Q\in G$ with $A\subseteq\conf(P,Q)$.
\end{proof}

\begin{lemma}
\label{lem:zeta-identity}
For every finite set $C$,
\[
     \sum_{A\subseteq C}(-2)^{|A|}=(-1)^{|C|}.
\]
\end{lemma}

\begin{proof}
Let $r=|C|$. The number of subsets $A\subseteq C$ of size $t$ is $\binom{r}{t}$. Therefore
\[
     \sum_{A\subseteq C}(-2)^{|A|}
     =\sum_{t=0}^{r}\binom{r}{t}(-2)^t = (1-2)^r=(-1)^r=(-1)^{|C|}.
\]
\end{proof}

\begin{theorem}
\label{thm:query-value}
Let $G$ be any multiset of previously inserted phase-free Pauli strings, and let $P$ be any query string. The query value
\[
      \anti_G(P)=\frac{|G|-Z_G(P)}{2},
      \qquad
      Z_G(P)=\sum_{A\subseteq\supp(P)}(-2)^{|A|}F_G(P,A),
\]
is exactly the number of strings $Q\in G$ that anticommute with $P$.
\end{theorem}

\begin{proof}
By Lemma~\ref{lem:FG-meaning}, for every $A\subseteq\supp(P)$,
\[
      F_G(P,A)=|\{Q\in G:A\subseteq\conf(P,Q)\}|.
\]
Thus
\[
\begin{aligned}
      Z_G(P)
      &=\sum_{A\subseteq\supp(P)}(-2)^{|A|}
        |\{Q\in G:A\subseteq\conf(P,Q)\}| \\
      &=\sum_{A\subseteq\supp(P)}(-2)^{|A|}
        \sum_{Q\in G}{\bf 1}[A\subseteq\conf(P,Q)],
\end{aligned}
\]
where ${\bf 1}[E]$ is $1$ if $E$ is true and $0$ otherwise. Since all sums are finite, exchange the order of summation:
\[
      Z_G(P)=\sum_{Q\in G}\sum_{A\subseteq\supp(P)}(-2)^{|A|}{\bf 1}[A\subseteq\conf(P,Q)].
\]
The inner sum only receives contributions from subsets $A\subseteq\conf(P,Q)$, because $\conf(P,Q)\subseteq\supp(P)$. Hence
\[
      Z_G(P)=\sum_{Q\in G}\sum_{A\subseteq\conf(P,Q)}(-2)^{|A|}.
\]
By Lemma~\ref{lem:zeta-identity}, the inner sum equals $(-1)^{|\conf(P,Q)|}$. Therefore
\[
      Z_G(P)=\sum_{Q\in G}(-1)^{|\conf(P,Q)|}.
\]
By Lemma~\ref{lem:conflict-sign}, $Q$ anticommutes with $P$ exactly when $|\conf(P,Q)|$ is odd. If $Q$ commutes with $P$, its contribution to $Z_G(P)$ is $+1$. If it anticommutes, its contribution is $-1$. Let $c$ be the number of strings in $G$ that anticommute with $P$. Then $|G|-c$ strings commute with $P$, so
\[
      Z_G(P)=(|G|-c)-c=|G|-2c.
\]
Solving for $c$ gives
\[
      c=\frac{|G|-Z_G(P)}{2}=\anti_G(P).
\]
Thus the query returns the desired anticommutation count.
\end{proof}

\begin{theorem}
\label{thm:count-return}
The procedure \textsc{CountAllAnticommutingPairs} returns
\[
  A\big(P^{(1)},\ldots,P^{(m)}\big)
   =\left|\left\{\{i,j\}:1\leq i<j\leq m,
      \ P^{(i)}P^{(j)}=-P^{(j)}P^{(i)}\right\}\right|.
\]
\end{theorem}

\begin{proof}
At the start of iteration $i$, the dictionary contains exactly the previously processed multiset
\[
      G_i=\big\{P^{(1)},\ldots,P^{(i-1)}\big\},
\]
because the algorithm inserts $P^{(j)}$ once at the end of iteration $j$ and inserts nothing else. By Theorem~\ref{thm:query-value}, the query value $c_i$ computed at iteration $i$ is exactly the number of indices $j<i$ such that $P^{(i)}$ anticommutes with $P^{(j)}$.

Every unordered pair $\{j,i\}$ with $j<i$ is considered exactly once, namely during iteration $i$. Therefore the final accumulated sum
\[
      T=\sum_{i=1}^{m}c_i
\]
counts each unordered anticommuting pair once and counts no commuting pair. This is exactly $A\big(P^{(1)},\ldots,P^{(m)}\big)$.
\end{proof}

\begin{corollary}
The certification variant returns ``all commute'' if and only if all input strings commute pairwise. If it returns a witness pair $(i,j)$, then $P^{(i)}P^{(j)}=-P^{(j)}P^{(i)}$.
\end{corollary}

\section{Running time, space, and optimality}

We use the following computational model. The sparse input stores each non-identity letter as a pair consisting of a qubit index and one of three labels. Dictionary lookup and update for a key of length $r$ costs $O(r)$ expected time, as with a standard hash table.

\begin{lemma}[Insertion cost]
\label{lem:insertion_cost}
Inserting a Pauli string $Q$ of weight $w$ performs $2^w$ dictionary updates and takes $O(w2^w)$ expected time. The total number of stored key occurrences over all insertions is at most $\sum_i2^{\wt\!\big(P^{(i)}\big)}$.
\end{lemma}

\begin{proof}
The insertion procedure enumerates all subsets $A\subseteq\supp(Q)$. A set of size $w$ has $2^w$ subsets. For each subset, the algorithm forms the key $(A,Q|_A)$, whose length is $|A|\leq w$, and increments its dictionary count. Thus the update time is
$
     O\!\left(\sum_{A\subseteq\supp(Q)}|A|\right).
$
Each of the $w$ support positions belongs to exactly half of the $2^w$ subsets, so
$
     \sum_{A\subseteq\supp(Q)}|A|=w2^{w-1}=O(w2^w).
$
The number of update operations is $2^w$. Summing $2^{\wt\!\big(P^{(i)}\big)}$ over all inserted strings gives the claimed bound on stored key occurrences before aggregation of duplicate keys. Aggregation can only reduce the number of distinct dictionary keys.
\end{proof}

\begin{lemma}[Query cost]
\label{lem:query-cost}
For a query string $P$ of weight $w$, \textsc{AntiCountAgainstPrevious} performs $3^w$ dictionary lookups and takes $O(w3^w)$ expected time.
\end{lemma}

\begin{proof}
For every subset $A\subseteq\supp(P)$, the algorithm enumerates all assignments $a\in\calL_P(A)$. Since each position in $A$ has two choices, $|\calL_P(A)|=2^{|A|}$. Hence the total number of dictionary lookups is
$
      \sum_{A\subseteq\supp(P)}2^{|A|}
      =\sum_{t=0}^{w}\binom{w}{t}2^t
      =(1+2)^w
      =3^w.
$
The key length in a lookup for subset $A$ is $|A|\leq w$, so the expected time is at most $O(w)$ per lookup, for total $O(w3^w)$.
\end{proof}

\begin{theorem}[Total complexity]
For input strings $P^{(1)},\ldots,P^{(m)}$ of weights $w_i=\wt\!\big(P^{(i)}\big)$, the counting algorithm takes expected time
$
      O\left(\sum_{i=1}^{m} w_i3^{w_i}\right)
$
and uses space
$
      O\left(\sum_{i=1}^{m} w_i2^{w_i}\right)
$
for dictionary keys plus the input storage. In particular, if $w_i\leq k$ for every $i$, the expected time is $O(mk3^k)$ and the additional space is $O(mk2^k)$.
\end{theorem}

\begin{proof}
At iteration $i$, the algorithm performs one query for $P^{(i)}$ and one insertion for $P^{(i)}$. By Lemma~\ref{lem:query-cost}, the query time is $O(w_i3^{w_i})$. By Lemma~\ref{lem:insertion_cost}, the insertion time is $O(w_i2^{w_i})$. Therefore the total expected time is
$
      O\left(\sum_{i=1}^{m} w_i3^{w_i}\right).
$
The dictionary receives at most $2^{w_i}$ keys from $P^{(i)}$, each of length at most $w_i$, so the additional space for stored keys and counts is
$
      O\left(\sum_{i=1}^{m}w_i2^{w_i}\right).
$
If all $w_i\leq k$, we obtain $O(mk3^k)$ time and $O(mk2^k)$ space.
\end{proof}

\begin{corollary}
The certification algorithm with witness returns either a correct all-commuting certificate or a correct anticommuting witness in expected time $O(mk3^k)$.
\end{corollary}

\begin{corollary}
For every fixed $k$, the locality-zeta algorithm is optimal up to a constant depending on $k$ in the sparse-input model.
\end{corollary}

If a downstream solver needs the entire dense edge list, then no algorithm can avoid quadratic output time. If the downstream task only needs to know whether a proposed family is commuting, how many anticommuting violations it contains, or whether there exists a violation, the locality-zeta algorithm avoids building the graph.

\section{Example}

Consider three two-qubit strings
\[
      P^{(1)}=\X\Y,
      \qquad
      P^{(2)}=\Y\Z,
      \qquad
      P^{(3)}=\Y\Id.
\]
For $P^{(1)}$ and $P^{(2)}$, the conflict set is $\{1,2\}$ because $\X$ differs from $\Y$ on qubit $1$ and $\Y$ differs from $\Z$ on qubit $2$. The size is even, so $P^{(1)}$ and $P^{(2)}$ commute. For $P^{(1)}$ and $P^{(3)}$, the conflict set is $\{1\}$, so they anticommute. For $P^{(2)}$ and $P^{(3)}$, the conflict set is empty on qubit $1$ because both have $\Y$, and qubit $2$ is ignored because $P^{(3)}$ has identity. Hence they commute.

Suppose $P^{(1)}$ and $P^{(2)}$ have been inserted and we query $P^{(3)}=\Y\Id$. The support is $S=\{1\}$. The query has two subsets: $A=\varnothing$ and $A=\{1\}$. For $A=\varnothing$, $F\big(P^{(3)},\varnothing\big)=2$ because there are two previous strings. For $A=\{1\}$, assignments conflicting with $\Y$ are $\X$ and $\Z$. Among previous strings, only $P^{(1)}$ has $\X$ on qubit $1$. No previous string has $\Z$ on qubit $1$. Thus $F\big(P^{(3)},\{1\}\big)=1$. The zeta sum is
$
      Z= (+1)\cdot 2 + (-2)\cdot 1 =0.
$
The anticommutation count against previous strings is
$
      (|G|-Z)/2=(2-0)/2=1,
$
correctly detecting that $P^{(3)}$ anticommutes with exactly one previous string.

\section{Conclusion}

The locality-zeta algorithm turns exact anticommutation counting for sparse $k$-local Pauli strings from a pairwise problem into a pattern-counting problem. It aggregates all previously seen local subpatterns and uses inclusion--exclusion to recover the parity of the conflict set against a query. The resulting time is linear in the number of Pauli strings and input-optimal for fixed locality $k$.
Pauli commutation is a ubiquitous inner loop in measurement grouping, Pauli arithmetic, and local Hamiltonian workflows. We anticipate that these results will be instrumental in scaling the classical subroutines necessary for future high-performance quantum simulation and error mitigation architectures.

\bibliographystyle{unsrt}
\bibliography{references}

\end{document}